\documentclass[letterpaper, 10pt, conference]{ieeeconf}  

\IEEEoverridecommandlockouts                              
\usepackage{bbm}
\usepackage{amssymb, amsfonts, amsmath, bm, balance}
\usepackage{tikz} 
\usetikzlibrary{positioning}

\newcommand{\abold}{\bm{\alpha}}
\newcommand{\ncal}{\mathcal{N}}
\newcommand{\mcal}{\mathcal{M}}
\newcommand{\rcal}{\mathcal{R}}
\newcommand{\mm}{\mcal^{-}}

\newcommand{\uhat}{\hat{U}}
\newcommand{\uopt}{U^*}
\newcommand{\pinew}{\bar{\pi}}
\newcommand{\rmax}{r_{\text{max}}}
\newcommand{\E}{\mathbb{E}}
\newcommand{\sbar}{\bar{s}}

\newtheorem{theorem}{Theorem}
\newtheorem{lemma}{Lemma}

\newtheorem{definition}{Definition}

\newtheorem{assumption}{Assumption}

\title{\LARGE \bf Corrected: On Confident Policy Evaluation for Factored Markov Decision Processes with Node Dropouts
}

\author{Carmel Fiscko$^{1\dagger}$, Soummya Kar$^{1}$, and Bruno Sinopoli$^{2}$
\thanks{$^{1}$Carmel Fiscko and Soummya Kar are with the Dept. of Electrical and Computer Engineering at Carnegie Mellon University in Pittsburgh, PA. {\tt\small cfiscko@andrew.cmu.edu, soummyak@andrew.cmu.edu}}%
\thanks{$^{2}$Bruno Sinopoli is with the Dept. of Electrical and Systems Engineering at Washington University in St. Louis, MO. {\tt\small bsinopoli@wustl.edu }}%
\thanks{$^{\dagger}$This material is based upon work supported by the National Science Foundation Graduate Research Fellowship Program under Grant No. DGE1745016. Any opinions, findings, and conclusions or recommendations expressed in this material are those of the author(s) and do not necessarily reflect the views of the National Science Foundation. Additional support provided by the Hsu Chang Memorial Fellowship in ECE.}%
}

\begin{document}

\maketitle
\thispagestyle{empty}
\pagestyle{empty}

\begin{abstract}

In this work we investigate an importance sampling approach for evaluating policies for a structurally time-varying factored Markov decision process (MDP), i.e. the policy's value is estimated with a high-probability confidence interval. In particular, we begin with a multi-agent MDP controlled by a known policy but with unknown transition dynamics. One agent is then removed from the system - i.e. the system experiences node dropout - forming a new MDP of the remaining agents, with a new state space, action space, and new transition dynamics. We assume that the effect of removing an agent corresponds to the marginalization of its factor in the transition dynamics. The reward function may likewise be marginalized, or it may be entirely redefined for the new system. Robust policy importance sampling is then used to evaluate candidate policies for the new system, and estimated values are presented with probabilistic confidence bounds. This computation is completed with no observations of the new system, meaning that a safe policy may be found before dropout occurs. The utility of this approach is demonstrated in simulation and compared to Monte Carlo simulation of the new system.

\end{abstract}

\section{Introduction}

Finding optimal control policies is a key goal of reinforcement learning (RL). While simple systems are easy to consider, the ``curse of dimensionality" means that as the state and action spaces grow, it takes longer and longer to evaluate the goodness of a candidate policy. It is therefore important to investigate methods that accelerate learning, such as leveraging systems' internal structure to reduce the scope of the problem.

One such method to handle complex state spaces is the factored or transition-independent MDP \cite{guestrin2003efficient}, in which the state and/or action spaces are expressed as the Cartesian product of some smaller set. This structure extends into the transition kernel, in which each state-action to state transition may be factored across the state space definition. This structure reduces the scope of the state and policy spaces, therefore accelerating the learning process and aiding efficient exploration \cite{osband2014near}. One application in which factored MDPs arise is in describing some multi-agent MDPs or stochastic games \cite{guestrin2001multiagent}, \cite{yang2020overview}, \cite{bucsoniu2010multi}. The MDP formulation may accomplish centralized control of a multi-agent system. If each agent moves within some set of individual states, the system-wide state may be described as the collection of individual agent states, i.e. the joint action of the agents. 

An issue arising in MDPs that describe natural processes (e.g. intelligent agents) is that many theoretical guarantees rely on an assumption of time-invariance. Fewer methods have tackled the case when the transition matrix may change over time \cite{papoudakis2019dealing}, \cite{ornik2019learning}. For example, the transition matrix may be assumed to be fixed over some time window, and an optimal control policy for that time window may be found. In the factored regime with an exponentially growing state space, this model-based approach may be intractable. The model-free regime gives a more flexible approach to time-varying systems, as estimated Q-functions may be constantly updated for the newest data based to application-specific heuristics. This technique lacks strong theoretical guarantees, however, as a necessary condition for convergence requires infinite observations of all state-action pairs \cite{sutton2018reinforcement}. 

We consider a different notion of time-varying: changes to the agent structure, specifically agent dropout. This phenomenon happens when, after some time has passed under normal operation, an agent leaves the system. For example, consider a group of computers executing a distributed sensing or learning scheme. At any moment a node may lose internet connectivity, may be corrupted due to an attack, or may be physically damaged. In a financial network, new businesses emerge and others go out of business, and in social media users constantly come and go.

Node dropout is often viewed as link failure in a communication graph between agents. In distributed settings where agents work together to learn a value function, update protocols have been developed that are robust to such link failures and stochastic networks \cite{kar2013cal}, \cite{zhang2018fully}. A general necessary assumption is that the graph is fully connected on average, which ensures convergence of the agents' value estimates. This work focuses on a related but different formulation: we consider a dynamically changing set of agents, where agents are assumed to have a long-term notion of joining or leaving the group. In this case, the assumption of average full connectivity over time will be violated. We therefore consider adding or removing an agent as the initialization of a \emph{new} MDP with a re-defined the state space, action space, reward function, and transition dynamics.

In this case of structural change, the central controller will want to know how to best act on the \emph{new} system, having only observed the \emph{old} system. If agents are added to the system, their transition behavior will be unknown so observations will be needed to form a policy. In the case of node dropout, however, the controller should be able to leverage knowledge of the old system in finding a policy for the new system. The reward function of the new system may reflect a similar goal to the original system, or it may be completely redefined. For example, the system may activate a ``safety mode," in that the controller wants high confidence of a minimum value after dropout occurs. While the new system can be solved with the aforementioned techniques for time-varying MDPs, it will be costly in time to restart learning from a fresh initialization, and may be unsafe to perform exploration in the new space. 

The goal of this paper is to perform high-confidence policy evaluation for the post-dropout system by using samples produced by the pre-dropout system. We assume that the post-dropout MDP retains structure reflective of its pre-dropout counterpart, including that the new transition kernel marginalizes the removed agent from the original transitions. We evaluate candidate policies by using importance sampling (IS) techniques on trajectories produced by the original system \cite{sutton1999between}, but with sample returns mapped to the reward function of the new system. The innovation here is that IS is used to evaluate policies for a \emph{different} MDP from the one that generated the observed data. This differs from standard policy IS techniques which are used to evaluate policies for the \emph{same} MDP that generated the data. Estimation techniques for two structures of reward function are then presented: (1) rewards marginalized with respect to the removed agent, and (2) rewards newly defined in the reduced agent space. Confidence bounds for the two methods are established, and tradeoffs between sample trajectory length and number of needed trajectories are discussed. 

Section \ref{setup} discusses the preliminaries of the multi-agent system setup. Section \ref{method} describes policy importance sampling, and discusses how to apply the estimators to the transformed MDP. Section \ref{analysis} presents the error bound of the estimator, and a simulation is shown in Section \ref{exp}.

\section{Problem Setup} \label{setup}
\subsection{Preliminaries}
Consider a set of agents (factors) $\mathcal{N}=\{1,\dots,N\}$. Let there be a factored MDP described by the tuple $\mathcal{M}=(\mathcal{S}, \mathcal{A}, R, T, \gamma)$. The state space factors across the agents as $\mathcal{S}=S_1\times\dots\times S_N$ where one state describes one combination of substates across all agents $s=[s_1,\dots,s_N]$. 
Any state with a set subscript as in $s_{\mathcal{X}}$ references the elements of the state $s$ indexed by the set $\mathcal{X}$; for example, $s_n$ means the substate of agent $n$ within the greater state. The action space is likewise factored $\abold\in \mathcal{A}=A_1\times\dots\times A_N$ where one action describes the action assigned by the central controller to each agent $\abold=[\alpha_1,\dots,\alpha_N]$. The reward function is defined to give a bounded deterministic scalar reward for each state-action pair. Furthermore the reward function satisfies the separable structure $r(s,\abold)=\sum_{n\in\ncal} r_n(s_n,\alpha_n)$, which corresponds to a summation across agent-specific rewards. Next, the transition kernel defines the state-action to state transition probabilities, and is assumed to be time-homogeneous and satisfy the factored Markovian form $P(s'|s,\abold) = \prod_{n\in\ncal} P(s_n'|s,\alpha_n)$. This factorization structure means that each agent has an independent transition when conditioned on the current state and action, and each agent is only dependent on the action assigned to it from the central controller. We will assume that the transition kernel satisfies the factored structure, but the exact probabilities will be unknown to the central controller. 
 The final component $\gamma\in(0,1)$ is a scalar discount factor that diminishes future rewards.

The goal of an MDP is to solve for a control policy $\pi: \mathcal{S}\to\mathcal{A}$. The \emph{value} of a policy $\pi$ is the expected reward,
\begin{equation}
    V^{\pi}_H(s) \triangleq\mathbb{E}_{\pi}\left[ \sum_{t=1}^{H}\gamma^{t-1}r(s_t,\abold_t)\vert s_0=s\right],
\end{equation}
where $H$ is a time horizon. The value for an infinite time horizon is attained by $V^{\pi}(s)\triangleq\lim_{H\to\infty} V^{\pi}_H(s)$. An optimal policy for an infinite time horizon, therefore, is a maximizer $\pi^*\in\text{argmax}_{\pi}V^{\pi}(s)$. As we consider a factored MDP, the policy can be independently factored across the agents given a state \cite{fiscko2022cluster}.

\subsection{Dropout}
The controller acts on $\mcal$ with some policy until the system experiences \emph{node dropout}, whereby agent $N$ leaves the system. The new MDP with agents $\ncal^{-}=\{1,\dots,N-1\}$ is referred to $\mm$. The updated state and action spaces become $S_1\times\dots\times S_{N-1}$ and $A_1\times\dots\times A_{N-1}$. A single state will be denoted by $\sbar$ and an action by $\abold^-$. Note that any agent may be removed, but the index $N$ is assumed without loss of generality (WLOG). 

As to not conflate the policies for the original and new systems, the policy for $\mcal$ will be denoted as $\pi$, and the policy for $\mm$ will be $\phi$. Values of the new system will likewise be denoted as $U$. The rest of this section will discuss the transition kernel and reward function for $\mm$.


Specifying the transition kernel of the new system is a key step as it directly informs the new optimal policy. A simple option is to pretend $s_N$ is fixed at its last observed value for all time; this simply uses the transition matrices of the original system with the fixed $s_N$, and a policy may be found through standard RL techniques. While this is a straight-forward method, this assumed transition structure will not be the universal best choice, as fixing an agent's supposed action is not equivalent to not observing the agent at all.

In this paper, it will be assumed that the effect of agent dropout on the transition probabilities will be equivalent to a marginalization of the agent. To make this assumption, we first need to assert existence of a unique stationary distribution to ensure the marginalization is well-defined. 

\begin{assumption} \label{ergodic}
Systems $\mcal$ and $\mm$ are both ergodic Markov chains under any fixed policy. 
\end{assumption}

Ergodicity ensures the uniqueness of the stationary distribution, and is a standard assumption for estimation of Markov chains.

\begin{assumption}\label{marg}
The state-action-state transition probabilities of system $\mm$ are equal to the marginalization of agent $N$ from the transition probabilities of $\mcal$ under $\pi$.
\end{assumption}
\begin{equation}
    P(s_n'|\sbar,\abold^-_n)=\E_{s_N}[P(s_n'|s,\abold_n)]. \label{agent marg}
\end{equation}
Note that due to the independent agent factorization of the transition matrix, it can be shown that \eqref{agent marg} implies,
\begin{equation}
    P(\sbar'|\sbar,\abold^-) = \mathbb{E}_{s_N, s_N',\alpha_N}[P(s'|s,\abold)].
\end{equation}

Assumption \ref{marg} describes applications in which agents account for the removal of agent $N$ by averaging out the effect of its substate. Note that the expectation is taken with respect to the stationary distribution of agent $N$'s substate, which in turn is dependent on the exerted policy $\pi$. This assumption intuitively means that the other agents absorb the effects of agent $N$ based on their overall observations of its behavior. An example of node dropout is shown in Figure \ref{dbn} drawn as a truncated dynamic Bayesian network (DBN) \cite{guestrin2001multiagent}. By this formulation node dropout can also be thought of as a time-varying connectivity graph between agents. 

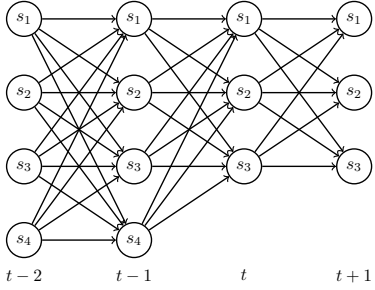
\begin{figure}
\centering
 \resizebox{.6\linewidth}{!}{
\begin{tikzpicture}[node distance={15mm}, thick, main/.style = {draw, circle}] 
\node[main] (1) {$s_1$}; 
\node[main] (2) [below of=1] {$s_2$}; 
\node[main] (3) [below of=2] {$s_3$}; 
\node[main] (4) [below of=3] {$s_4$};
\node[draw=none] (9) [below = 0.1cm of 4] {$t-2$};
\node[main] (5) [right = 1.5cm of 1] {$s_1$}; 
\node[main] (6) [right = 1.5cm of 2] {$s_2$};
\node[main] (7) [right = 1.5cm of 3] {$s_3$};
\node[main] (8) [right = 1.5cm of 4] {$s_4$};
\node[draw=none] (10) [below = 0.1cm of 8] {$t-1$};

\node[main] (11) [right = 1.5cm of 5] {$s_1$}; 
\node[main] (12) [right = 1.5cm of 6] {$s_2$};
\node[main] (13) [right = 1.5cm of 7] {$s_3$};
\node[draw=none] (14) [right = 1.5cm of 8] {};
\node[draw=none] (15) [below = 1.6cm of 13] {$t$};
\node[main] (16) [right = 1.5cm of 11] {$s_1$}; 
\node[main] (17) [right = 1.5cm of 12] {$s_2$};
\node[main] (18) [right = 1.5cm of 13] {$s_3$};
\node[draw=none] (19) [right = 1.5cm of 14] {};
\node[draw=none] (20) [below = 1.6cm of 18] {$t+1$};

\draw[->] (1) -- (5); 
\draw[->] (1) -- (6); 
\draw[->] (1) -- (7); 
\draw[->] (1) -- (8); 
\draw[->] (2) -- (5); 
\draw[->] (2) -- (6); 
\draw[->] (2) -- (7); 
\draw[->] (2) -- (8); 
\draw[->] (3) -- (5); 
\draw[->] (3) -- (6); 
\draw[->] (3) -- (7); 
\draw[->] (3) -- (8); 
\draw[->] (4) -- (5); 
\draw[->] (4) -- (6); 
\draw[->] (4) -- (7); 
\draw[->] (4) -- (8); 

\draw[->] (5) -- (11); 
\draw[->] (6) -- (11); 
\draw[->] (7) -- (11); 
\draw[->] (8) -- (11); 
\draw[->] (5) -- (12); 
\draw[->] (6) -- (12); 
\draw[->] (7) -- (12); 
\draw[->] (8) -- (12); 
\draw[->] (5) -- (13); 
\draw[->] (6) -- (13); 
\draw[->] (7) -- (13); 
\draw[->] (8) -- (13); 

\draw[->] (11) -- (16); 
\draw[->] (11) -- (17); 
\draw[->] (11) -- (18); 
\draw[->] (12) -- (16); 
\draw[->] (12) -- (17); 
\draw[->] (12) -- (18); 
\draw[->] (13) -- (16); 
\draw[->] (13) -- (17); 
\draw[->] (13) -- (18); 
\end{tikzpicture} }
\caption{This example network demonstrates dropout of agent 4 at time $t$. The set of times $\{\dots, t-2, t-1\}$ is the pre-dropout era and $\{t, t+1,\dots\}$ is the post-dropout era. The arrows represent the structure of the agent-specific transitions, i.e. $P(s_n'|s_1, s_2, s_3, s_4, \alpha)$ pre-dropout and $P(s_n'|s_1, s_2, s_3, \alpha)$ post-dropout.}
\label{dbn}
\end{figure}

The final step is to update the definition of the reward function. In this work, we assume that the central controller chooses one of the following reward regimes:


\textbf{Reward Option 1: Marginalized.} If the controller wishes to preserve the reward structure of $\mcal$ when moving to $\mm$, the reward function can be updated according to Assumption \ref{marg}. The new reward function may be defined as the marginalization of $s_N$, where  $\bar{r}(\sbar,\abold^-)=\mathbb{E}_{s_N,\alpha_N}[r(s,\abold)]$.

This reward structure may be used when the reward function is inherent to the problem, and the controller wishes to preserve the specified goodness of states. This may also be useful when the controller wants to set a general rule to account for the new reward function, which may be updated automatically regardless of which agent drops out.

\textbf{Reward Option 2: Refreshed.} The controller may want to define a wholly new reward function when moving to the new MDP. For example, the MDP may enter a ``safety mode" if it loses an agent, and wants to control for different behavior until the agent may be re-added. We still assume that the rewards are additive across the new agent set $\mathcal{N}^-$, deterministic, and finite. The final reward function for $\mm$ will be denoted as $\bar{r}(\sbar,\abold^-)$, and the mapping $\mathcal{R}$ will translate $r(s,\abold)\to \bar{r}(\sbar, \abold^-)$.

\subsection{Dataset}
The last assumption, and last part of the problem setup, describes the type of data available to the controller. While observing the system, the controller collects sample trajectories of the form $(s,\abold, s')$ produced under $\pi$ to form the dataset $D$. Here the policy may be the optimal $\pi^*$, but that is not a requirement.

\begin{assumption}\label{known}
Assume access to a dataset $D$ consisting of $|D|$ trajectories of system $\mcal$ produced by policy $\pi$. The dataset is of the form $(s,\abold,s')_t^i$ where $t$ indexes time and $i$ refers to the trajectory. The specific generating transition kernel, i.e. $P(s'|s,\abold)$, is not assumed to be known.  
\end{assumption}



\section{Method}\label{method}
Given the problem setup, the objective is to find a good policy $\phi$ for $\mm$ by analyzing the data produced by $\mcal$. The controller would desire that the proposed policy $\phi$ will yield a minimum value with high confidence. This computation must be completed without any observations of $\mm$, as the new system should be safely controlled at its inception.

The general method to find a good policy is to evaluate a set of policies and select the best performer. As suggested by \cite{thomas2015high}, the ``best performer" may be defined as the policy with the best lower confidence bound on its estimated return. Our focus is thus to evaluate proposed policies, and to return a high confidence error bound on their values. 

First, note that under assumptions  \ref{ergodic}, \ref{marg}, and the reward structures, the following relationship may be drawn between the value functions of the two MDPs:


\begin{lemma} The value of system $\mm$ under policy $\pi^-=\E_{s_N,\alpha_N}\pi(\abold|s)$ with marginalized rewards is, \label{values}
\begin{equation}
    U^{\pi^-}(\sbar) = \E_{s_N,\alpha_N}[V^{\pi}(s)].
\end{equation}
\end{lemma}
\begin{proof}
Note that by the Bellman operator and iterated expectation,
\begin{align*}
    V^{\pi}(s) &=\E_{\abold}\E_{s'}\left[ r(s,\abold)+\gamma V(s') \right],\\
    &=\E_{\abold^-}\E_{\sbar'}\left[\E_{s_N,s_N',\alpha_N}[ r(s,\abold)+\gamma V^{\pi}(s')|\abold^-,\sbar] \right],\\
    \begin{split}
        &=\E_{\abold^-}\E_{\sbar'}\big[\E_{s_N,\alpha_N}[r(s,\abold)|\abold^-,\sbar]\\
        &\quad\quad +\gamma\E_{s_N,s_N',\alpha_N}[V^{\pi}(s')|\abold^-,\sbar] \big],
    \end{split}
\end{align*}
where the outer expectations are taken with respect to $\abold^-\sim \pi(\abold^-|s)$ and $\sbar'\sim P(\sbar'|s,\abold)$. Then note that $p(s_N,s_N',\alpha_N)=P(s_N'|s_N,\alpha_N)p(s_N,\alpha_N)$ and since $P(s_n'|s,\alpha_n)$ are independent,
\begin{align*}
    &\sum_{s_N,s_{N}',\alpha_N}P(s'|s,\abold)p(s_N, s_N',\alpha_N), \\
    &=\sum_{s_N,s_{N}',\alpha_N}P(\sbar'|s,\abold)P(s_N'|s_N,\alpha_N)p(s_N,\alpha_N),\\
    &=\sum_{s_N,\alpha_N}P(\sbar'|s,\abold)p(s_N,\alpha_N)\sum_{s_{N}'}P(s_N'|s_N,\alpha_N),\\
    &=P(\sbar'|\sbar,\abold^-).
\end{align*}
Then substituting $V^{\pi}(s)\to \E_{s_N,\alpha_N}V^{\pi}(s)$ will marginalize the effects of $\alpha_N$ and $s_N$ from the policy:
\begin{align*}
&\E_{s_N,\alpha_N}V^{\pi}(s) = \sum_{\abold^-}\E_{s_N,\alpha_N}\pi(\abold|s)\big[\E_{s_N,\alpha_N}[r(s,\abold)|\abold^-,\sbar]\\
&+\gamma \sum_{\sbar'}P(\sbar'|\sbar,\abold^-)\E_{s_N',\alpha_N'}V^{\pi}(s')\big],
\end{align*}
which can be compared to the Bellman operator to verify equivalence to the value equation for $U^{\pi^-}(\sbar)$. 

\end{proof}
Note that under refreshed rewards, $\E_{s_N,\alpha_N}\bar{r}(\sbar,\abold^-)=\bar{r}(\sbar,\abold^-)$, so an equivalent result may be found by mapping $r\to\bar{r}$ before marginalization of the value function and noting that infinite applications of the Bellman operator converge to the true value.

\subsection{Limitations of the Bellman Optimality Conditions}
Given Lemma \ref{values}, a natural question is if this relationship extends to value optimality. Consider a system $\mm$ constructed with marginalized rewards. The first issue is that $\uopt(s^-)$ cannot be readily calculated as $\mathbb{E}_{s_N,\alpha_N}[V^*(s)]$. To see this, note that the Bellman optimality criterion necessitates that $V^*(s) = \mathbf{T}V^*(s) = \max_{\pi} \mathbb{E}_{\pi}[r(s,\abold)+\gamma V^*(s')]$. Therefore, if $\tilde{U}$ are optimal values, then they must satisfy:
\begin{align}
    \tilde{U}(\sbar)&=\mathbf{T}\tilde{U}(\sbar),\nonumber\\
    &=\mathbf{T}\E_{s_N,\alpha_N}V^*(s),\nonumber\\
    &=\max_{\pi^-}\E_{\pi^-}[r(\sbar,\abold^-)+\gamma \E_{s_N,\alpha_N}V^*(s')]. \label{bellman}
\end{align}

However, we find that,
\begin{align}
    \tilde{U}(\sbar)&=\E_{s_N,\alpha_N}[V^*(s)],\nonumber\\
    &=\E_{s_N,\alpha_N}\E_{\pi^*}[r(s,\abold)+\gamma V^*(s')],\nonumber\\
    &=\E_{s_N,\alpha_N}[\max_{\pi}\E_{\pi}[r(s,\abold)+\gamma V^*(s')]]. \label{margeq}
\end{align}

\normalsize
Clearly, \eqref{bellman} and \eqref{margeq} are not guaranteed to coincide as the maximization and marginalization operations do not commute. There are no guarantees that the Bellman optimality condition still holds in the new MDP space; this calculation has merely evaluated the value of the \emph{new} system under the policy $\{\pi^*_1,\dots,\pi^*_{N-1}\}$ developed for the \emph{old} system. This policy is not necessarily optimal for the \emph{new} system. Therefore, one possibility is to evaluate $V^{\pi}(s)$ for other policies $\pi$ and find one such that the principle of optimality holds for $\mm$. Solving for such a policy is difficult to do analytically, so another approach could be to evaluate several candidate policies and choose the best performing option. 

In practice, however, it is unlikely that $V^{\pi}$ has been solved for all $\pi$, as a user is most likely to have focused on $V^{\pi^*}$. In addition, this marginalization technique will not work if the reward function has been refreshed, as a new definition of the reward function will render the old $V^{\pi}$ useless. 

\subsection{Transformed Policy Importance Sampling}
The technique of policy importance sampling (IS) will be used to complete the stated objective, as it handles many of the issues raised in the previous section. Policy IS will be used to estimate infinite-horizon $U^{\phi}$ given a dataset of finite trajectories produced by system $\mcal$ under \emph{behavior policy} $\pi$. The goal is to estimate the returns of the trajectories had they been generated by $\mm$ with \emph{target policy} $\phi$. The return of a general MDP trajectory of length $H$ is defined as,

\small
\begin{equation}
    G_H \triangleq \sum_{t=1}^H \gamma^{t-1}r(s_t,\abold_t).
\end{equation}
\normalsize
By abuse of notation, the transformed reward for $\mm$ will be referred to as $\rcal(G_H)$ where,
\small
\begin{equation}
    \rcal(G_H) = \sum_{t=1}^H \gamma^{t-1} \rcal(r_t)= \sum_{t=1}^H \gamma^{t-1} \bar{r}_t.
\end{equation}
\normalsize

IS techniques were developed to evaluate new policies for the \emph{same} MDP that generated the observed samples, but in the node dropout problem we evaluate new policies on the \emph{transformed} system $\mm$. As in Assumption \ref{known}, let each trajectory $\tau=(s_1,\abold_1,r_1,\dots, s_H,\abold_H, r_H)$ be generated from MDP $\mcal$ under policy $\pi$. Consider direcly applying an IS technique to evaluate MDP $\mm$ under some new policy $\phi$. With $p$ as the joint distribution of $\tau$ under $\phi$ and $q$ the joint distribution of $\tau$ under $\pi$, the return can be expressed as,

\small{\begin{equation}
    v_H^{\phi} = \E_{\tau\sim q}\left[\frac{p(\tau)}{q(\tau)}\sum_{t=1}^H \gamma^{t-1} \rcal(r_t) \right].
\end{equation}}
\normalsize

Many estimators for $G_H^{\phi}$ have been proposed  \cite{sutton1999between} \cite{precup2000eligibility}, such as the per-trajectory IS estimator, per-step estimator, weighted estimator, and doubly robust estimator \cite{thomas2016data} \cite{jiang2016doubly}. However, these estimators cannot be applied directly because in the node dropout setting the state-action to state transition terms must be handled appropriately. Note that the ratio of $p$ and $q$ for some initial state distribution $d$ and deterministic reward is,

\small{\begin{equation}
    \frac{d(\sbar_0)\phi(\abold^-_1|\sbar_1)P(\sbar_2|\sbar_1,\abold_1^-)\dots \phi(\abold^-_H|\sbar_H)}{d(s_0)\pi(\abold_1|s_1)P(s_2|s_1,\abold_1)\dots \pi(\abold_H|s_H)},\label{ratiobad}
\end{equation}}
\normalsize

Note that the ratio in \eqref{ratiobad} cannot be evaluated as it is assumed that the transition matrix is unknown and therefore all terms $P(\sbar_{t+1}|\sbar_t,\abold^-_t)/P(s_{t+1}|s_t,\abold_t)$ are likewise unknown.

In comparison, the solution we propose is to augment policy $\phi$ with $\pi_N$ to form policy
\begin{equation}
    \phi'(\abold|s) = \phi(\abold^-|\sbar)\pi_N(\alpha_N|s_N)
\end{equation}
which is defined on pre-dropout MDP. Note that as we consider factored MDPs, then the policies for each agent (factor) may be considered independently \cite{fiscko2022cluster}. This means that $\phi$ and $\pi_N$ are independent policies and may be multiplied to yield $\phi'$. Then, the desired post-dropout policy $\phi(\abold^-|\bar{s}) = \mathbb{E}_{s_N,\alpha_N}\phi'(\abold|s)$ is simply the marginalization of the augmented policy $\phi'$. The value of policy $\phi'$ may now be analyzed via standard policy IS on trajectories generated by $\pi$ on the pre-dropout MDP, and the resulting estimated returns may be marginalized to remove $s_N$ from the system.

Under this scheme, the ratio of $p$ and $q$ becomes,

\small{\begin{equation}
    \frac{d(s_0)\phi'(\abold_1|s_1)P(s_2|s_1,\abold_1)\dots P(s_H|s_{H-1},\abold_{H-1})\phi(\abold_H|s_H)}{d(s_0)\pi(\abold_1|s_1)P(s_2|s_1,\abold_1)\dots P(s_H|s_{H-1},\abold_{H-1})\pi(\abold_H|s_H)}.
\end{equation}}
\normalsize

The transition and initial state terms will cancel such as in standard policy IS, leaving behind the importance ratio,
\begin{equation}
    \rho_t = \frac{\phi'(\abold_t|s_t)}{\pi(\abold_t|s_t)}.
\end{equation}
The following standard assumption on the generating policy is made:
\begin{assumption}\label{support}
$\phi'$ is fully supported on $\pi$, i.e. for $\abold$ such that $\phi'(\abold|s)>0$ then $\pi(\abold|s)>0$.
\end{assumption}

Note that Assumption \ref{support} refers to the generating policy, not to the observed dataset $D$. Given a finite number of samples it is likely that some state-action pairs may not be seen, but they must have a non-zero probability of occurrence. This assumption further implies that the data must not have been generated by a deterministic policy, but options like $\epsilon$-soft policies are permitted. Any standard policy IS estimator may then be used to estimate $v_H^{\phi'}$. 

The final step is to apply Lemma \ref{values} to transform the estimated value of $\phi'$ on the pre-dropout system to the estimated value of the desired policy $\phi$ on the post-dropout system. The final value for $U^{\phi}_H(s)$ may be calculated via marginalization,
\begin{equation}
    U^{\phi}_H(\sbar) = \E_{s_N,\alpha_N}[V_H^{\phi'}(s)],
\end{equation}
which in turn may be estimated as,
\begin{align}
    U^{\phi\dagger}_H(\sbar) &= \sum_{\alpha_N}\sum_{s_N} V_H^{\phi'}(s)\phi'(\alpha_N|s)\hat\mu(s_N),\nonumber\\
    &=\sum_{\alpha_N}\sum_{s_N} V_H^{\phi'}(s)\pi(\alpha_N|s)\hat\mu(s_N),\label{uhatdagger}
\end{align}
given an empirical estimate $\hat\mu(s_N)$ of the stationary state distribution. The dagger notation $\dagger$ will refer to an empirically marginalized return. 
\normalsize
\section{Analysis}\label{analysis}
The performance of the IS estimators on the transformed MDP $\mm$ are discussed in this section. For all results, Assumptions  \ref{ergodic}, \ref{marg}, \ref{known}, and \ref{support} are assumed to be satisfied.

\subsection{Marginalization}
To perform the final marginalization step, the stationary distribution $\mu(s_N)$ must be used. As this distribution is not assumed to be known, an empirical distribution $\widehat{\mu}(s_N)$ is instead used in \eqref{uhatdagger}. Under assumption \ref{ergodic} there will exist a unique stationary distribution, so $\widehat{\mu}(s_N)$ may be estimated as the number of specific substate occurrences divided by the number of total occurrences. 


To study the error of this method, concentration bounds are first stated for the empirical stationary distribution as developed by \cite{paulin2015concentration}. First it is necessary to define the \emph{mixing time}, which measures the time required by a
Markov chain for the distance to stationarity to be small \cite{levin2017markov}. 

\begin{definition}
With $ d(t) \triangleq \sup_{s\in\Omega}d_{TV}(P^t(s,\cdot),\mu)$ and $t_{\text{mix}}(\epsilon)\triangleq \min\{t:d(t)\leq \epsilon\}$, the  \emph{mixing time} is defined as:

\vspace{-4mm}
\small
\begin{gather}
    t_{\text{mix}}\triangleq t_{\text{mix}}(1/4).
\end{gather}
\end{definition}
\normalsize
Based on the mixing time, the convergence of the empirical stationary distribution may be bounded as follows:
\begin{lemma}\label{stationary} 2.19 from \cite{paulin2015concentration}. Let $\hat{\mu}(s) \triangleq \frac{1}{H}\sum_{t=1}^H \mathbf{1}[S_{t}=s]$ be the empirical distribution for state $S$. Let $d(D) \triangleq d_{TV}(\hat{\mu}(s), \mu)$ be the total variational distance between the empirical and true distributions. For any $\epsilon\geq 0$,

\small
\begin{equation}
    P(\lvert d(D)-\mathbb{E}[d]\rvert \geq \epsilon) \leq 2 \text{exp}\left(-\frac{\epsilon^2 H}{4.5t_{mix}}\right).
\end{equation}
\end{lemma}
\normalsize

This bound means that the empirical stationary distribution will converge to the true stationary distribution as the trajectory length goes to infinity. The more complicated the chain, i.e. the larger the mixing time, the more samples are needed for small errors to be achieved.

This bound may be adapted for the substate $s_N$ by using the mixing time associated with the single factor with state space $S_N$, action space $A_N$, and transition $P(s_N'|s,\alpha_N)$. The mixing time may then be bounded such as with equation 7.2 from \cite{levin2017markov}.

\subsection{Policy Evaluation}
Policy IS is used to estimate the return of policy $\phi'$ given trajectories generated by $\pi$. In our analysis no specific IS estimator is assumed, although tighter bounds may be achieved for specific estimators. Furthermore, note that the sample returns are built from finite horizon trajectories, but are used to estimate the infinite-horizon value. In the following analysis there are therefore four possible sources of error: error due to IS, error due to the finite time horizon, error due to the sample marginalization, and any bias induced by the choice of estimator. The following concentration bound may be stated.


\begin{theorem}\label{maintheorem}
Let $U^{\phi}(\sbar)$ be the true infinite horizon value of system $\mm$ under policy $\phi$. Let $\uhat^{\phi\dagger}_H(\sbar)$ be the sample average estimate of $U^{\phi}_H(\sbar)$ formed by IS and then marginalization. Let the IS estimator be constructed from $|D|-1$ i.i.d. trajectories each of length $H$ with bounded bias $|\E[\widehat{U}^{\phi}_H]-U^{\phi}_H|\leq B_{IS}(H)$. Let the empirical stationary distribution $\hat\mu(s_N)$ used to marginalize $\uhat^{\phi}_H(\sbar)$ be formed from the last trajectory in the dataset with length $H_{\mu}$. Let $\rmax\triangleq\max_{\sbar,\abold^-} \bar{r}(\sbar,\abold^-)$, and let $\epsilon'= \frac{\gamma^H}{1-\gamma}\rmax +B_{IS}(H)$. Then for $\delta\geq0$,

\small\begin{equation}
\begin{split}
    &P(|U^{\phi}(\sbar)-\uhat^{\phi\dagger}_H(\sbar)|\geq \delta+\epsilon')\\
    &\leq 2\exp\left(-\frac{(1-\gamma)^2H_{\mu}\delta^2}{9t_{mix}(1-\gamma^{H_{\mu}})^2r_{\text{max}}^2|S_N|^{2}}\right)\\
    &+ 2\exp\left(-\frac{(1-\gamma)^2(|D|-1)\delta^2}{2(1-\gamma^H)^2r_{\text{max}}^2}\right) .
    \end{split}\label{bound1}
\end{equation}

\end{theorem}

\normalsize
Theorem \ref{maintheorem} gives an overall exponential confidence interval for an estimator constructed from finite trajectories subject to IS and empirical marginalization. The first error term arises due to marginalization and the second due to the IS estimator. The IS error term will go to zero as $|D|\to\infty$, and the marginalization error will go to zero as horizon length $H_{\mu}\to\infty$. Note that having two different horizon lengths with $H \ll H_{\mu}$ is beneficial as the variance of standard IS estimators increases rapidly with $H$, but $H_{\mu}$ must be large enough for the empirical stationary distribution to reflect the true stationary distribution. 




\begin{proof}
Let $U^{\phi}_H$ be the value for finite time horizon $H$ and $\hat{U}^{\phi\dagger}_H$ be the value estimated by IS and then marginalized. The error can be decomposed as,

\small{
\begin{align*}
    &U^{\phi}(\sbar)-\hat{U}^{\phi\dagger}_H(\sbar),\\
     \begin{split}
    &=(U^{\phi}(\sbar)-U^{\phi}_H(\sbar))+(U^{\phi}_H(\sbar)-U^{\phi\dagger}_H(\sbar)),\\
    &+(U^{\phi\dagger}_H(\sbar)-\E\widehat{U}^{\phi\dagger}_H(\sbar))+(\E\widehat{U}^{\phi\dagger}_H(\sbar)-\widehat{U}^{\phi\dagger}_H(\sbar)),
    \end{split}\\
    &=\Delta_{H}+\Delta_{\mu}+B_{IS}+\Delta_{\text{IS}}.
\end{align*}}

\normalsize
\noindent where $\Delta_H$ is the error of a trajectory of length $H$ to estimate the infinite horizon reward, $\Delta_{\mu}$ is the error due to empirical marginalization, $B_{IS}$ is the upper bound of the (possible) bias of the IS estimator, and $\Delta_{\text{IS}}$ is the error of the selected IS estimator. For any value function $U^{\phi}$ it is known that,

\small
\begin{equation*}
    \lvert U^{\phi}(\sbar)-U^{\phi}_H(\sbar)\rvert \leq \rmax\gamma^H/(1-\gamma) = \tilde{r},
\end{equation*}

\normalsize
\noindent so $\Delta_H$ can be deterministically bounded. The stochastisity in the error therefore comes from the importance sampling step. With substitution, the triangle inequality, and noting that $\tilde{r}\geq 0$ and $B_{IS}\geq 0$,

\footnotesize\begin{align*}
&P(|\Delta_{\text{H}}+\Delta_{\mu}+B_{IS}+ \Delta_{\text{IS}}|\geq \epsilon)\\
    &\leq P(|\tilde{r}| + |B_{IS}|+ |\Delta_{\mu}+\Delta_{\text{IS}}|\geq\epsilon),\\
    &=P(|\Delta_{\mu}+\Delta_{\text{IS}}| \geq \epsilon-\tilde{r}-B_{IS}),\\
    &\leq P\left(|\Delta_{\mu}|\geq\frac{1}{2}(\epsilon-\tilde{r}-B_{IS})\right) + P\left(|\Delta_{IS}|\geq\frac{1}{2}(\epsilon-\tilde{r}-B_{IS})\right).
\end{align*}

\normalsize
Next, consider bounding $\Delta_{\mu}$.
\small\begin{align*}
    &|U^{\phi}_H(s)-U^{\phi\dagger}_H(s)| = \sum_{\alpha_{N}}\sum_{s_N}V^{\phi'}_H(s)\phi'(\alpha_N|s)|\mu(s_N)-\widehat{\mu}(s_N)|\\
    &\leq \max_{s_N}|\mu(s_N)-\widehat{\mu}(s_N)|\sum_{\alpha_{N}}\sum_{s_N}V^{\phi'}_H(s)\pi(\alpha_N|s).
\end{align*}
\normalsize
Therefore,
\small\begin{align*}
    &P(|\Delta_{\mu}|\geq \epsilon_0)\\
    &\leq P\left(d_{TV}\left(\mu,\widehat{\mu}\right)\geq \epsilon_0/\Big(\sum_{\alpha_{N}}\sum_{s_N}V^{\phi'}_H(s)\pi(\alpha_N|s)\Big)\right).
\end{align*}
\normalsize
With the assumption of unbiased estimators applied to i.i.d. trajectories, standard confidence inequalities may be used. The Hoeffding inequality, bounding $\max\widehat{U}^{\phi\dagger}_H(s)$ and $V_{H}^{\phi'}$, and rearranging terms complete the proof.
\end{proof}

In dealing with large state spaces where $|D|$ for each starting $s$ may be limited, first-visit Monte Carlo policy evaluation may be used to produce more unbiased sample returns. This method will draw on different trajectory lengths for each sample, requiring adjustments to the confidence interval. One technique may be to use first-visit or every-visit Monte Carlo only when the length of the trajectory after first visit exceeds some minimum amount. The confidence interval may then be calculated from the minimum allowed trajectory length $\bar{H}$.

A question here may be raised as to choose the best $\bar{H}$. Increasing $\bar{H}$ will exponentially reduce the error of $|U^{\phi}-U_H^{\phi}|$, but it will also exponentially increase the variance of most IS estimates. A design choice could therefore be selecting $\bar{H}$ to balance the bias and variance of the estimator. The variance term is comprised of variance of the Monte Carlo estimator which depends on the stochasticity of the trajectories, and the variance of the selected IS method. As upper bounds for variance of IS methods tend to be loose or formed via recursion, a closed form solution for $\bar{H}$ cannot be easily found. A natural solution is to save part of the training data to select an acceptable $\bar{H}$ by cross-validation. Sample trajectories may be truncated for different values of $H$, from which the MDP transformation and IS may be computed. The best $H$ may be used for the remaining training trajectories to estimate the final value. 

\section{Simulation} \label{exp}

This problem was simulated to evaluate multiple policies post-dropout. A system of $N$ agents was constructed with $|S_n|=3$ states possible for each agent. The pre-dropout system assigned an indicator reward  $r(s) = \sum_n \mathbf{1}(s_n=1)$ where each state was rewarded by the number of agents in substate 1. The post-dropout system instead used $\bar{r}(s) = \sum_n \mathbf{1}(s_n=2)$ which rewarded substate 2. 

As a ``proof of concept" to motivate the need for an improved policy post-dropout, see Figure \ref{optgap}, which compares the performance of policy $\E_{s_N,\alpha_N}\pi^*(\abold|s)$ to $\phi^*$ on the post-dropout system. Note that there is a large gap in received return between these two policies, showing that it is beneficial to solve for a new policy in the post-dropout regime. 

A dataset of the pre-dropout system under policy $\pi$ was collected where $|D|=1000$ and $H=500$. The doubly robust estimator was used to evaluate a randomly chosen policy, and the $95\%$ confidence bound was found. In Figure \ref{graph2}, the dashed line shows the estimated value of $\phi$, which closely approximates the displayed true value. In comparison, the Monte Carlo (MC) value estimate, averaged over $10$ trajectories, is formed by directly observing the new system under $\phi$. It can be seen that to achieve a MC estimate within the bounds of the DR estimate, trajectories of at least $H'\approx 245$ samples would be needed, with multiple trajectories desired to reduce variance. 

This simulation shows how importance sampling techniques can be applied to transformed returns to evaluate the goodness of policies for unobserved systems. The estimate was constructed from samples purely taken from $\mcal$, whereas any MC evaluation needed to observe the new system $\mm$. The returned confidence bound means that the central controller can exert policy $\phi$ and expect to receive an infinite-horizon value within the specified range; in MC evaluation there are no such guarantees, meaning the controller could cause risk exerting a dangerous policy on the system. In addition, this evaluation was done with a new reward function that shared no similarity to the original reward, allowing for great flexibility in application. 

\begin{figure}
    \centering
    \includegraphics[width=0.8\linewidth]{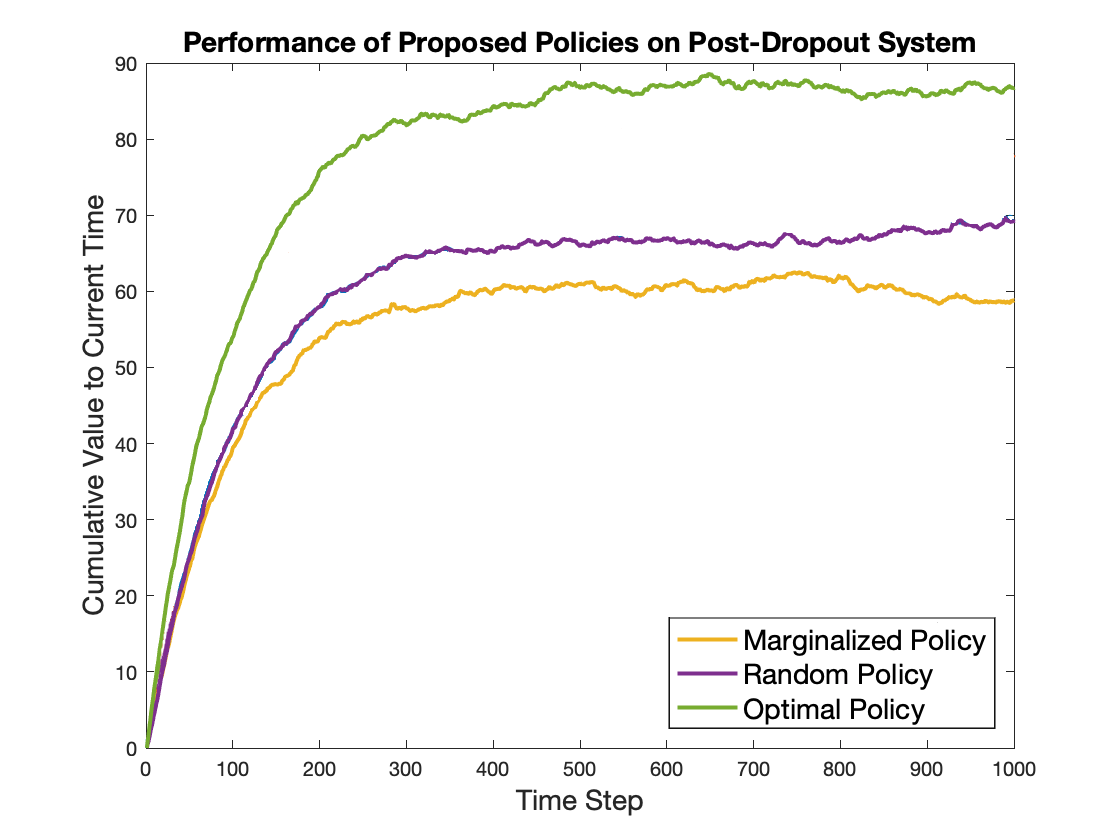}
    \caption{MC policy evaluation showing the performance of the marginalized policy $\E_{s_N,\alpha_N}\pi^*(\abold|s)$ that was optimal for $\mathcal{M}$, the true optimal policy $\phi^*$, and a randomly selected policy. Note the optimality gap; this motivates that $\E_{s_N,\alpha_N}\pi^*(\abold|s)$ may not suffice as a good policy in the case of dropout.  }
    \label{optgap}
\end{figure}
\begin{figure}
    \centering
    \includegraphics[width=0.8\linewidth]{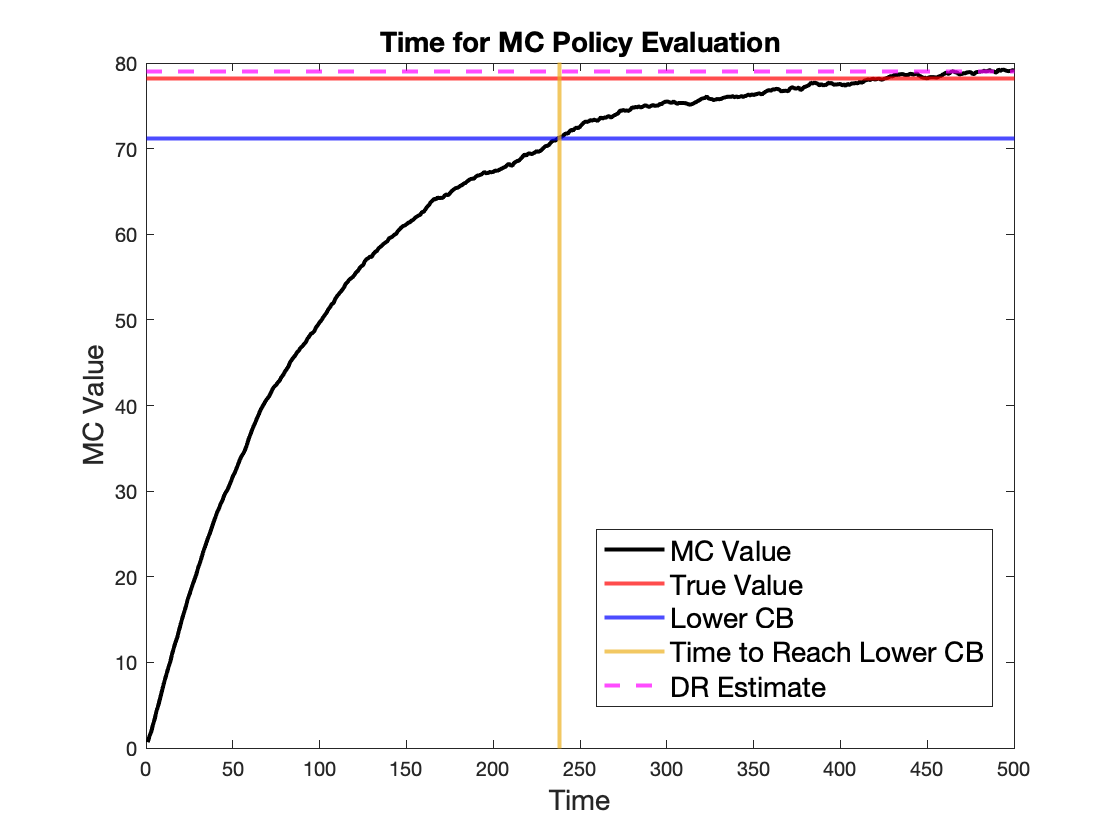}
    \caption{MC policy evaluation convergence shown versus the DR estimator and its lower $95\%$ confidence bound for a randomly chosen policy $\phi$. The true value was $78.21$, the estimated value was $79.01$, and the $95\%$ confidence bound was $\pm 6.98$. Note that the DR estimate is constructed from samples of $\mcal$, but MC policy evaluation needs samples from $\mm$. }
    \label{graph2}
\end{figure}

\section{Conclusion}
This work presented a method for high confidence policy evaluation for factored MDPs under the case of node dropout where the effects of the dropped agent are marginalized. This method allows policies to be evaluated for the post-dropout system using only observations of the pre-dropout system. We show that standard policy IS cannot be directly applied as the pre- and post- dropout systems are different MDPs, but the specific structural properties of this problem enable a modified version of IS to be used. We demonstrate how to modify policy IS for this setting, and present a confidence bound on the resulting estimator. This means that good policies for the post-dropout system maybe precomputed, or policies for the pre-dropout system may be analyzed for their robustness to dropout. Future work can build upon this work to accomplish the policy search step.




\balance
\bibliographystyle{IEEEtran}
\bibliography{IEEEabrv,tocite}

\addtolength{\textheight}{-5cm}
\end{document}